\documentclass[12pt]{article}
\usepackage{amsmath}
\usepackage{amsfonts}
\usepackage{amsthm}
\usepackage{times}
\usepackage[pdftex]{graphicx} 
\def\Z{{\mathbb Z}}

\def\R{{\mathbb R}}
\def\acts{\triangleright}
\def\T{{\mathbb T}}

\DeclareMathSymbol\crossrt{\mathrel}{AMSb}{"6E}
\DeclareMathSymbol\crosslt{\mathrel}{AMSb}{"6F}

\def\oh{\frac{1}{2}}

\def\tr{\hbox{Tr\ }}

\def\acts{\triangleright}

\newtheorem{lemma}{Lemma}[section]
\newtheorem{theorem}[lemma]{Theorem}

\title{On spectral action over Bieberbach manifolds.}
\author{
Piotr Olczykowski, 
Andrzej Sitarz\thanks{Partially supported by MNII grants 189/6.PRUE/2007/7 and N 201 1770 33} \\
{\em Institute of Physics, Jagiellonian University}\\
{\em Reymonta 4, 30-059 Krak\'ow, Poland} \\
}
\date{November 30, 2010}
\begin{document}
\maketitle
\begin{abstract}
We compute the leading terms of the spectral action for orientable three dimensional Bieberbach manifolds first, using two different methods: the Poisson summation formula and the perturbative expansion. Assuming that the cut-off function is not necessarily symmetric we find that that the scale invariant part of the perturbative expansion might differ from the spectral action of the flat three-torus by the eta invariant. 
\end{abstract}
\section{Introduction}

Bieberbach manifolds are compact manifolds, which are quotients of 
the Euclidean space by a free, properly discontinuous and isometric 
action of a discrete group. The classification of all Bieberbach manifolds 
is a complex problem in itself (see \cite{Sad} and the references therein). 
The non-trivial lowest dimensional examples appear in three dimensions and have 
been first described in \cite{Bieb1,Bieb2}. 

In this paper we study whether the nonperturbative form of the spectral 
action as introduced by Connes \cite{CoCh} and calculated for several examples
in \cite{CoCh2} and \cite{MaPi} can distinguish between the torus itself 
and the chosen Bieberbach three-manifolds.

\section{Three-dimensional Bieberbach Manifolds}

In this section, we briefly recall the description of three-dimensional 
Bieberbach manifolds as quotients of three-dimensional tori by the action 
of a finite discrete group. We use the algebraic language, starting with 
the algebra of the functions on the three-torus $\T^3$, which we view as 
generated by three mutually commuting unitaries $U,V,W$. 

There are, in total, 10 different Bieberbach three-dimensional manifolds, 
six orientable (and that includes the three-torus) and four nonorientable 
ones. They all could be defined as quotients of the three-torus by the 
following free actions of a discrete group on the three-torus. We
restrict our attention here to the orientable manifolds different from
the three-torus. The respective action of discrete groups on the unitary 
generators of the algebra of functions on the three-torus is summarized in 
the following table:

\begin{center}
\begin{tabular}{|c|c|c|l|}
\hline
name & group & generators & action on the generators of $\T^3$, \\ \hline \hline
G2& $\Z_2$ & e & $e \acts U = -U$, $e \acts V = V^*$, $e \acts W = W^*$ \\ \hline
G3& $\Z_3$ & e & $e \acts U =  e^{\frac{2}{3}\pi i} U$, $e \acts V = W*$, $e \acts W = W^* V$ \\ \hline
G4& $\Z_4$ & e & $e \acts U =  i U$, $e \acts V = W$, $e \acts W = V^*$ \\ \hline
G5& $\Z_6$ & e & $e \acts U =  e^{\frac{1}{3}\pi i} U$, $e \acts V = W$ , $e \acts W = W V^*$ \\ \hline
G6& $\Z_2 \times \Z_2$ & $e_1,e_2$ & $e_1 \acts U =  -U$, $e_1 \acts V = V^*$, $e_1 \acts W = W^*$ \\
& & & $e_2 \acts U =  U^*$, $e_2 \acts V = -V$, $e_2 \acts W = -W^*$ \\ \hline
\end{tabular}
\end{center}

\section{The spectra of the Dirac operator}

The spectrum of the Dirac operator on Bieberbach manifolds has been first
calculated by Pf\"affle \cite{Pfa}. We shall use his results, though rewritten
in a more form more suitable for our purposes. As the covering three-torus we 
shall choose the equilateral one (with lengths of three fundamental circles
equal). 

The spectrum can be also easily derived using the formalism of real 
spectral triples (which we shall present elsewhere \cite{OlSi2}). 

Let us fix here the notation. By $D^3_\tau$ we denote the Dirac operator
on the three-torus with equal lengths of three circles, with the eigenvalues:
$$ \lambda_{k,l,m} = \pm \sqrt{(k+\epsilon_1)^2 + |l+\epsilon_2 + \tau(m+\epsilon_3)|^2},
\;\; k,l,m \in \Z, $$
where $\epsilon_1$, $\epsilon_2$, $\epsilon_3$ are $0$ or $\oh$ and depend on the choice
of the spin structure and $|\tau|=1$, $\tau$ not real. The choice of $\tau=i$ corresponds
to the usually assumed Dirac operator.
We shall denote the spectrum of $D^3_\tau$, counted with the multiplicities, 
by ${{\cal S}p}^3_\tau$. Further, we shall need the generalized Dirac operator on the 
circle, taking the standard one, with eigenvalues:
$$ \lambda_{k} = \alpha k+\beta, \;\;\; k \in \Z, \alpha,\beta \in \R, $$
we shall denote its spectrum by ${{\cal S}p}^1_{\alpha,\beta}$. 

Whenever we write a coefficient in front of the spectrum set we mean the same
set but with the multiplicities reduced by that factor (of course, if the
coefficient is $\frac{1}{n}$ this requires that the multiplicities must 
be divisible by $n$.

We shall parametrize the spins structures of the three-torus by 
$\{\epsilon_1,\epsilon_2,\epsilon_3\}$ being $0$ or $\oh$, 
additionally, the representation of the discrete group on the 
Hilbert space of spinors can add additional spin structure, 
denoted by \hbox{$\delta=\pm 1$}.

We have the following spectra of the Dirac operator:

\begin{itemize}
\item $G2$

Here we must have $\epsilon_k=\oh$ and have eight possible spin structures, 
parametrized by choice of $\epsilon_l,\epsilon_m$ and $\delta$. As the 
Dirac operator on three torus we take $D^3_i$, with eigenvalues 
$\pm \sqrt{(k+\oh)^2+(l+\epsilon_2)^2+(m+\epsilon_3)^2}$, $k,l,m \in \Z$.

$$ 
{{\cal S}p}_{G2} = 
\begin{cases}
\frac{1}{2} {{\cal S}p}^3_i & \hbox{if\ } \epsilon_2=\oh \hbox{\ or\ } \epsilon_3=\oh \\
\frac{1}{2} \left( {{\cal S}p}^3_i \setminus 2 {{\cal S}p}^1_{1,\oh} \right)
\cup 2 {{\cal S}p}^1_{2,\oh-\delta} & \hbox{if\ } \epsilon_2=\epsilon_3=0, \delta=\pm 1. 
\end{cases}
$$

Observe that only in the $\epsilon_l=\epsilon_m=0$ case the spectrum is not 
the same as for the torus.

\item $G3$

In this case only the spin structures with $\epsilon_2=\epsilon_3=0$ could be projected
to the quotient space. The parameter $\delta$ is fixed by the choice of the spin 
structure $\epsilon_k$. As the projectable Dirac we take $D^3_{e^{\frac{2\pi i}{3}}}$. 

$$ 
{{\cal S}p}_{G3} = 
\begin{cases}
\frac{1}{3} \left( {{\cal S}p}^3_{e^{\frac{2\pi i}{3}}} \setminus 2 {{\cal S}p}^1_{1,\oh} \right)
\cup 2 {{\cal S}p}^1_{3,\oh} & \hbox{if\ } \epsilon_1=\oh, \delta=1, \\
\frac{1}{3} \left( {{\cal S}p}^3_{e^{\frac{2\pi i}{3}}} \setminus 2 {{\cal S}p}^1_{1,0} \right)
\cup 2 {{\cal S}p}^1_{3,-1} & \hbox{if\ } \epsilon_1=0,\delta=-1. 
\end{cases}
$$

\item $G4$

For the action of $Z_4$ only $\epsilon_1=\oh$ and $\epsilon_2=\epsilon_3$ spin structures 
could be projected onto the quotient, the Dirac operator which commutes with the action of
the discrete group is $D^3_{e^{\frac{\pi i}{4}}}$. There are four possible spin structures
and the corresponding spectra are:

$$ 
{{\cal S}p}_{G4} = 
\begin{cases}
\frac{1}{4} \left( {{\cal S}p}^3_{e^{\frac{\pi i}{4}}} \right) & \hbox{if\ } \epsilon_2=\epsilon_3=\oh, \\
\frac{1}{4} \left( {{\cal S}p}^3_{e^{\frac{\pi i}{4}}} \setminus 2 {{\cal S}p}^1_{1,\oh} \right)
\cup 2 {{\cal S}p}^1_{4,\frac{3}{2}-\delta} & \hbox{if\ } \epsilon_2=\epsilon_3=0,\delta= \pm 1. 
\end{cases}
$$

\item $G5$

Here the situation is similar as in the $G_3$ case and only the spin structures 
with $\epsilon_1=\oh$ and $\epsilon_2=\epsilon_3=0$ could be projected to the 
quotient space. The parameter $\delta$ is free and gives us two possible
spin structures. The projectable Dirac we take $D^3_{e^{\frac{2 \pi i}{3}}}$. 

$$ 
{{\cal S}p}_{G5} = 
\begin{cases}
\frac{1}{6} \left( {{\cal S}p}^3_{e^{\frac{2\pi i}{3}}} \setminus 2 
{{\cal S}p}^1_{1,\oh} \right)
\cup 2 {{\cal S}p}^1_{6,\oh} & \hbox{if\ } \epsilon_k=\oh, \delta = 1. \\
\frac{1}{3} \left( {{\cal S}p}^3_{e^{\frac{2\pi i}{3}}} \setminus 
2 {{\cal S}p}^1_{1,\oh} \right)
\cup 2 {{\cal S}p}^1_{6,\frac{7}{2}} & \hbox{if\ } \epsilon_k=\oh, \delta = -1.
\end{cases}
$$

\item $G6$

In this case the only projectable spin structure are those with 
$\epsilon_1=\epsilon_2=\epsilon_3=\oh$, the projectable Dirac 
operator is $D^3_i$ and the spectrum of $D$ remains the same 
(apart from rescaled multiplicities) for each of four spin 
structures over $G6$.

$$ 
{{\cal S}p}_{G6} = {{\cal S}p}^3_{i}.
$$
\end{itemize}

\section{The spectral action}

Since we have split the spectra of the Dirac operators into the sets, which
corresponds to the known cases, we shall begin by calculating the spectral
action of the corresponding three-dimensional and one-dimensional tori.

We assume here that the spectral action depends on $D$ and not on $D^2$ 
(that is we do not restrict ourselves to the even functions over 
the spectrum), therefore there is a slight change of notation when 
compared to the \cite{MarPie}. 

We begin with the action for the torus with the Dirac $D^3_\tau$:

$$ 
\begin{aligned}
{\cal S}(D^3_\tau,\Lambda) 
&= 2 \sum_{k,l,m} f \left(\pm \frac{\sqrt{(k+\epsilon_k)^2 
                  + |l+\epsilon_l + \tau(m+\epsilon_m)|^2}}{\Lambda} \right) \\
  &=  \widehat{f_e}(0,0,0) + o(\Lambda^{-1}) 
  = 2 \int dx \int dy \int dz f \left(\frac{\sqrt{x^2 + |y + \tau z|^2}}{\Lambda} \right)\\
  &= \frac{8 \pi^2}{\sin \phi} \Lambda^3 \int_0^\infty d\rho\, f_e(\rho) \rho^2 +
     o(\Lambda^{-1}). 
\end{aligned}
$$
where $\tau=e^{i\phi}$, $\widehat{f}$ denotes the Fourier transform of $f$ considered
as a function of three variables:
$$ \widehat{f}(k_x,k_y,k_z) = \int_{\R^3} f(x,y,z) e^{2\pi i (k_x x+ k_y y+ k_z z)} dx dy dz,$$
and we have used the Poisson summation formula.
 
For the Dirac operator $D^1_{\alpha,\beta}$ over the circle we have:
$$ 
\begin{aligned}
{\cal S}(D^1_{\alpha,\beta},\Lambda) 
&= \sum_{k} f \left( \frac{\alpha k+\beta}{\Lambda} \right) \\
  &= \hat{f}(0) + o(\Lambda^{-1}) 
  =  \Lambda \int_{\R} \left( \frac{\alpha k+\beta}{\Lambda} \right) dx 
    + o(\Lambda^{-1}) \\
  &= \frac{1}{\alpha} \Lambda \int_{\R} f(x) dx  + o(\Lambda^{-1}). 
\end{aligned}
$$
In the formula above $\hat{f}$ is the usual Fourier transform of $f$.
Let us observe that the following identity occurs:

$$ {\cal S}(D^1_{1,\gamma},\Lambda) = {\alpha}{\cal S}(D^1_{\alpha,\beta},\Lambda).$$
independently of the values of $\alpha,\beta$ and $\gamma$.

We can now formulate:

\begin{theorem}
The nonperturbative spectral action over the orientable Bieberbach manifolds 
with the Dirac operator projected from the equilateral Dirac operator over the 
three-torus is (up to an scaling and order $o(\Lambda^{-1})$ indistinguishable 
from the spectral action over the three-torus.
\end{theorem}

\begin{proof}
Of course, only the cases when the spectrum differs significantly for the spectrum
of the Dirac over three torus may give rise to some differences. However, observe 
that the difference in the spectra is always of the form:
$$ - \frac{1}{n} 2 \left( {{\cal S}p}^1_{1,\gamma} \right) + 2 {{\cal S}p}^1_{n,\beta} $$
where the sign denotes whether we substract the spectrum (which would result in 
substracting the resulting component of the spectral action) or add it.
The constants $\beta$ and $\gamma$ vary from case to case, $n$ depends on the 
order of the discrete group $G$ so that the manifold if $T^3/G$.

{}From the observation above, however, we know that the resulting spectral action 
components will not depend on $\beta$ and $\gamma$ and we will obtain:
$$ - \frac{2}{n} \left( \int_R f(x) dx \right) + \frac{2}{n} \left( 
\int_R f(x) dx \right) = 0, $$
and hence will not contribute to the leading terms of the spectral action.
\end{proof}

\section{The perturbative spectral action}

We can explicitly calculate the difference between the spectral 
action on the Bieberbach manifolds and the spectral action on 
the three-torus, expanding then the result in $\Lambda$.

Consider first an even function $f$. Taking $f$ be a Laplace transform 
of $h$:
$$ f(s) = \int_0^\infty e^{-sx} h(x) dx, $$
we can write:
$$ 
\tr f(\frac{|D|}{\Lambda}) = \int_0^\infty \tr e^{-x \frac{|D|}{\Lambda}} h(x) dx.
$$

Knowing the spectra of the Dirac operator over Bieberbach manifolds when 
compared to the three torus we can calculate the difference. 

First we need the technical lemma. Let $D^1_{\alpha,\beta}$ be (as denoted before) 
the spectrum of the rescaled Dirac over the circle. We calculate
exactly the exponential $ \tr e^{-p |D^1_{\alpha,\beta}|}$, assuming that 
$|\beta|<\alpha$.

$$ 
\begin{aligned}
\tr e^{-p |D^1_{\alpha,\beta}|} &= 
e^{-p \beta} + \sum_{k=1}^\infty e^{-p(\alpha k +\beta)} + \sum_{k=1}^\infty e^{-p (\alpha k -\beta)} \\
&= e^{-p |\beta|} + \left( \frac{e^{-p \alpha}}{1- e^{-p\alpha}} \right) 2 \cosh (p\beta).
\end{aligned}
$$

Taking into account that $p = \frac{x}{\Lambda}$ we can take the Laurent expansion for
large values of $\Lambda$:

\begin{equation}
e^{-p |\beta|} + \left( \frac{e^{-p \alpha}}{1- e^{-p\alpha}} \right) 2 \cosh (p\beta)
\sim \frac{2}{\alpha} \frac{\Lambda}{x} + o(\frac{x}{\Lambda}). 
\label{exp1}
\end{equation}

We can now state:

\begin{theorem}
The even component of the function $f$ in the spectral action is the same up to order 
$o(\Lambda^{-1})$ on all three-dimensional Bieberbach manifolds (including three-torus).
\end{theorem}

\begin{proof}
As we have seen only for some of the spin structures the spectra of the Dirac operator differ
from the spectrum of the three torus (apart from the trivial factor of multiplicities). In general,
we have, for the Bieberbach $G_x$, which is a quotient of the torus $T^3$ the following relation
of the spectra:
$$ {{\cal S}p}(G_x) = \frac{1}{n_x} \left( {{\cal S}p}(T^3) 
                     \setminus 2 {{\cal S}p}_{1,\epsilon} \right) 
                      + 2 {{\cal S}p}_{n_x,\epsilon'}, $$
where $n_x$ are integer constants, which are $n_2=2,n_3=3,n_4=4,n_5=6$ and $\epsilon,\epsilon'$ 
depend on the choice of spin structure.

So the difference between the spectral actions on the three-torus $T^3$ and on the 
Bieberbach $G_x$ could be calculated from this difference of the spectra:

\begin{equation}
{\cal S}(G_x,\Lambda) - \frac{1}{n_x} {\cal S}(T^3,\Lambda) = 
2 \sum_{\lambda \in {{\cal S}p}_{n_x,\epsilon'}} f\left(\frac{\lambda}{\Lambda} \right)
-\frac{2}{n_x} \sum_{\lambda \in {{\cal S}p}_{1,\epsilon}} f\left(\frac{\lambda}{\Lambda} \right).
\end{equation}

As the spectra ${{\cal S}p}_\pm$ are the spectra of rescaled Dirac on the circle, using
the result \ref{} we see that only the $\Lambda$ component in the perturbative expansion
could appear. Calculating it explicitly:

$$
{\cal S}(G_x,\Lambda) - \frac{1}{n_x} {\cal S}(T^3,\Lambda) = 
\Lambda \left( 2 \frac{2}{n_x} - \frac{2}{n_x} 2 \right) \int \frac{1}{x} f(x) dx + o(\Lambda^{-1})
=  o(\Lambda^{-1}).
$$

Therefore, irrespective of the chosen spin structure and Bieberbach manifold, even component
of the function determining the spectral action gives the same result.
\end{proof}

The situation is different for the odd component of $f$. We can always write, for an odd
function $f$:
$$ f\left( \frac{D}{\Lambda} \right) = \frac{D}{|D|} \phi \left( \frac{|D|}{\Lambda} \right),$$
where $\phi$ is an even function. Assuming that $\phi$ is a Laplace transform of $h$ the odd part 
of the spectral action becomes:
$$ 
\tr f\left( \frac{D}{\Lambda} \right) = \tr  \frac{D}{|D|} \phi \left( \frac{|D|}{\Lambda} \right)
= \int_0^\infty \tr \frac{D}{|D|} e^{-x \frac{|D|}{\Lambda}} h(x) dx.
$$

For the spectra of Dirac operators, which we know, we can calculate the function under
the integral:

$$ \tr \frac{D}{|D|} e^{-x \frac{|D|}{\Lambda}} $$

and obtain (again we denote $p=\frac{x}{\Lambda}$:

$$ 
\begin{aligned}
\tr \hbox{sign}(D^1_{\alpha,\beta}) e^{-p |D^1_{\alpha,\beta}|} &= 
\hbox{sign}(\beta) e^{-p |\beta|} + \sum_{k=1}^\infty e^{-p(\alpha k + \beta)} - \sum_{k=1}^\infty e^{-p (\alpha k -\beta)} \\
&= \hbox{sign}(\beta) e^{-p |\beta|} - \left( \frac{e^{-p \alpha}}{1- e^{-p\alpha}} \right) 2 \sinh (p\beta).
\end{aligned}
$$

We can expand the function for small $p$ around $0$:
$$
 \hbox{sign}(\beta) e^{-p |\beta|} - \left( \frac{e^{-p \alpha}}{1- e^{-p\alpha}} \right) 2 \sinh (p\beta)
 \sim \hbox{sign}(\beta) \frac{\alpha - 2|\beta|}{\alpha} + o(p). 
$$

Therefore, only (up to terms of order $o(\Lambda^{-1})$ only scale invariant 
term can appear. We have:

\begin{theorem}
The odd component of the function $f$ gives rise to a difference in the spectral action
on the Bieberbach manifolds in the scale invariant part of the action. The difference 
equals the eta-invariant of the Dirac operator on the Bieberbach manifold. 
\end{theorem}

\begin{proof}
First of all, observe that for the rescaled Dirac operator on the circle 
$D_{\alpha,\beta}$ the term:
$$ \hbox{sign}(\beta) \frac{\alpha - 2|\beta|}{\alpha}, $$
is the eta invariant $\eta(D^1_{\alpha,\beta})$, which measures the antisymmetry 
between the positive and negative parts of the spectrum of $D^1_{\alpha,\beta}$. 
Therefore, for any of the spin structures of the circle, the term vanishes for
the standard Dirac operator (that is, $D^1_{1,\oh}$ or $D^1_{1,0}$, using the 
notation of the paper). As a consequence, the difference between the 
(rescaled) spectral action on the three-torus $T^3$ and on the Bieberbach 
$G_x$ is (up to order $o(\Lambda^{-1})$:

$$
{\cal S}(G_x,\Lambda) - \frac{1}{n_x} {\cal S}(T^3,\Lambda) = 
2 \eta(D^1_{n_x,\epsilon'_x}) \phi(0),  
$$
where $n_x$ is as before and $\epsilon'_x$ depends on the chosen spin 
structure, and we have used that $\phi$ is a Laplace transform of 
$h$, so that:
$$ \int_0^\infty h(x) dx = \phi(0).$$
As this is, however, the only component of the spectrum of 
the Bieberbach manifolds, we have:
$$ 2 \eta(D^1_{n_x,\epsilon'_x}) = \eta(D^3_{G_x,\epsilon}), $$
and, finally:
$$ {\cal S}(G_x,\Lambda) = \eta(D^3_{G_x,\epsilon}) \Phi(0).$$ 

The value of $\eta$ invariant can be calculated explicitly for the 
manifolds $G2$,$G3$,$G4$,$G5$ and the chosen spin structures for 
which it does not vanish, the following table shows the results
(A and B denote the spin structures, giving rise to asymmetric 
 spectrum, in the order presented in section (3)):

\begin{center}
\begin{tabular}{|c|c|c|}
\hline
name & A & B \\ \hline \hline
G2& $1$ & $-1$  \\[1mm] \hline
G3& $\frac{4}{3}$ & $-\frac{2}{3}$  \\[1mm] \hline
G4& $\frac{3}{2}$ & $-\frac{1}{2}$  \\[1mm] \hline
G5& $\frac{5}{3}$ & $-\frac{1}{3}$  \\[1mm] \hline
\end{tabular}
\end{center}
\end{proof}

In fact, the result is not entirely surprising. From the general 
results of Bismut and Freed \cite{BiFr} one knows the small-$t$ asymptotic 
of the following function of the Dirac operator on the odd dimensional 
manifolds:
$$ \tr \frac{D}{|D|} e^{-t|D|} = \eta(D) + \sum_{l=0}^\infty 
(A_l +B_l \log t) t^{2l+2}. $$

We can calculate then the leading term of the spectral action arising from
an odd function to be:

$$ {\cal S}(D,\Lambda,f_o) = \eta(D) \phi(0) + o(\Lambda^{-1}). $$

We shall finish this section by observing why this effect was not 
picked by the methods used earlier, which involved sum over the 
entire spectrum with the help of the Poisson summation formula.

Observe that the $\eta$ invariant would appear if $\phi(0) \not= 0$. Since
our function $f(x) = \hbox{sign}(x)\phi(|x|)$ that means that $f$ is odd, 
but discontinuous at $x=0$. Therefore, the previous considerations were valid
but since were (implicitly) assuming continuity of $f$ we could not have 
obtained any deviation from the spectral action over the torus.

\section{Conclusions}

We have shown that apart from the possible difference arising from the
eta invariant the perturbative spectral action is exactly the same for all three dimensional Bieberbach manifolds as for the three torus. This is not at all surprising as all terms in the perturbative expansion (for the symmetric cut-off) depend on the Riemann curvature and Bieberbach manifolds are flat. The new result is the appearance of slight modifications when the cut-off function has an asymmetric part. 
Although here we obtain an invariant, it would be interesting to see if such a term might appear in some more complicated models, with some extra degrees of freedom coming from discrete spectral triples, for instance.


\end{document}